\documentclass[submission,copyright,creativecommons]{eptcs}

\providecommand{\event}{AFL 2026} 

\usepackage{iftex}

\ifpdf
  \usepackage{underscore}         
  \usepackage[T1]{fontenc}        
\else
  \usepackage{breakurl}           
\fi

\usepackage{amsmath,amssymb,amsthm,mathtools}
\usepackage{enumitem}
\usepackage{comment}
\usepackage{doi}
\usepackage{lineno}

\newcommand{\bigo}{\mathcal{O}}

\title{Subsequence Analysis Problems for Binary Parikh Matrices}
\author{Szil\'ard Zsolt Fazekas\institute{Akita University}\footnote{This work was supported by JSPS KAKENHI Grant Number JP26K14700.}\email{szilard.fazekas@ie.akita-u.ac.jp} \and Xinhao Huang\institute{Akita University}\email{d8523004@s.akita-u.ac.jp} \and Robert Merca\c {s}\institute{Loughborough University}\email{R.G.Mercas@lboro.ac.uk}}
\def\titlerunning{Subsequence Analysis Problems for Binary Parikh Matrices}
\def\authorrunning{Fazekas, Huang, Merca{\c s}}
\date{}

\newtheorem{theorem}{Theorem}[section]
\newtheorem{lemma}[theorem]{Lemma}
\newtheorem{corollary}[theorem]{Corollary}
\newtheorem{proposition}[theorem]{Proposition}

\newcommand{\Subseq}{\operatorname{Subseq}}

\newcommand{\iu}{\iota}
\newcommand{\iexists}{\iota_{\exists}}
\newcommand{\iforall}{\iota_{\forall}}

\begin{document}
\maketitle


\begin{abstract}
The {\em universality index} corresponding to a word is the largest integer \(k\) such that every word of length \(k\) over the given alphabet occurs in it as a subsequence. Relative to languages this notion can be investigated with respect to both {\em existential} and {\em universal} quantifiers, with the former corresponding to the existence of a word in the language with universality index at least \(k\), while the latter considers the index across all words that the language contains. 

In this work we study the existential ($\exists$-universality) and universal ($\forall$-universality) subsequence universality (as introduced in [Adamson et al., ISAAC 2023]) for binary languages consisting of all words with a fixed number of letters \(a\), letters \(b\), and subsequences \(ab\).  We prove that both
$\exists$- and $\forall$-universality admit exact arithmetic characterizations
for these languages.

Extending the above notions, we end the paper by initiating the analysis of the probability of a fixed word occurring as subsequence of the words in such a language. To this end, we prove that, for every fixed pattern word $w$ and an error tolerance, given as input the subsequence counts describing a language, the ratio of words in the language having the pattern $w$ as a subsequence admits an additive approximation scheme that is polynomial-time in the binary encoding of the input counts.
\end{abstract}

\section{Introduction}

A word \(u\) is a subsequence of a
word \(w\) if \(u\) can be obtained from \(w\) by deleting letters without changing
the order of the remaining letters.  This relation is central in combinatorics on
words and formal language theory, not least because of the classical theorem of
Higman, which states that finite words over a finite alphabet are well-quasi-ordered
by the subsequence relation~\cite{Higman52}, and a closely related theorem of Haines~\cite{Haines69} that states that the downward closure of every language is regular.
These results have made subsequences a
natural tool in the study of languages, well-quasi-orders, and verification problems:
even when a language itself is complicated, the set of all subsequences occurring
in its words admits a finite automata-theoretic representation.  This observation
has been used repeatedly in algorithmic formal language theory and in the analysis
of infinite-state systems, where downward closures and subword orderings provide
regular over-approximations and decidability tools~\cite{Zetzsche15,Zetzsche16,KNS16}.

Beyond the existence of finite representations for downward closures, one may ask
more refined algorithmic questions about the subsequences occurring in words or in
languages.  For a single word, this includes subsequence matching (decidable in linear time), 
testing whether two words have the same subsequences of at most a given length, i.e., testing Simon’s congruence~\cite{Simon75,KKS15,GawrychowskiEtAl21}, identifying absent subsequences~\cite{KoscheKossManeaSiemer22}, and studying universality with
respect to subsequences.  A word \(w\) over an alphabet \(\Sigma\) is called
\(k\)-subsequence universal, or simply \(k\)-universal, if every word in
\(\Sigma^k\), i.e., every word over $\Sigma$ of length $k,$ occurs in \(w\) as a subsequence.  Equivalently, the set of
subsequences of \(w\) contains the full \(k\)-spectrum \(\Sigma^k\).  This notion
and several variants have been investigated for individual words and compressed
words, e.g., see the work on scattered factor universality
\cite{dlt2020,Adamson2025105357} and the references therein.

Recently, the same family of questions has been studied for formal languages rather
than individual words.  In this setting, the input language is typically represented
succinctly, e.g., by an automaton or a grammar, and one asks whether a given
pattern occurs as a subsequence in some word of the language, in all words of the
language, or whether the language contains a word that is \(k\)-universal.  These
existential and universal variants have been considered for language classes from
the Chomsky hierarchy, including regular, context-free, and context-sensitive
languages~\cite{Adamson2025105357,FazekasKMMS24}.  In particular, the transition
from individual words to languages naturally introduces two questions: an
existential one, asking whether some word in the language has the desired
subsequence property, and a universal problem, asking whether all words in the
language have it.

In this paper we study these subsequence analysis problems for a different,
arithmetically defined class of finite languages.  We work over the binary alphabet
\(\{a,b\}\).  For nonnegative integers \(A,B,C\), let
\(
L_{A,B,C}
\)
be the language of all binary words with the numbers of
letters \(a\), letters \(b\), and scattered subwords \(ab\) equal $A$, $B$ and $C$, respectively.  Equivalently, this
is the equivalence class determined by a binary Parikh matrix~\cite{Prodinger1979,mateescu2001sharpening}, and also by the
binary \(2\)-deck~\cite{Kalashnik73,Scott97}, since over a binary alphabet the count of length-one
subsequences and that of \(ab\), characterizing the former, determines the whole count of length-\(2\) subsequences, which describes the latter.
The language is nonempty if and only if \(0\le C\le AB\), and throughout the paper
the entries \(A,B,C\) are part of the input, encoded in binary.

Our first goal is to determine the existential and universal subsequence
universality indices of \(L_{A,B,C}\).  For a language \(L\), we say that \(L\) is
\(k\)-$\exists$-universal if some word \(w\in L\) is \(k\)-universal, and we
write
\[
\iota_{\exists}(L)
=
\max\{k \mid \text{some } w\in L \text{ is } k\text{-universal}\}.
\]
We say that a nonempty language \(L\) is \(k\)-$\forall$-universal if
every word \(w\in L\) is \(k\)-universal, and~we~write
\[
\iota_{\forall}(L)
=
\max\{k \mid \text{every } w\in L \text{ is } k\text{-universal}\}.
\]
The notions, introduced in~\cite{Adamson2025105357}, capture the two natural quantifier choices over the language:
the largest $k$ such that a $k$-universal representative exists, and the largest $k$ such that all representatives
are $k$-universal.

We prove that, for binary \(2\)-deck classes, 
both indices are computable in
polynomial time in the bitlength of \(A,B,C\).
In the existential case, we prove that 
\(
\iota_{\exists}(L_{A,B,C})
\)
is the minimum among $A$, $B$ and a simple formula depending on the three input multiplicities.
On the universal side, we show that no nonempty language
\(L_{A,B,C}\) is \(4\)-universal; thus only the patterns of lengths
one, two, and three need to be considered. A complete avoidance analysis for the
eight binary patterns of length three then yields a characterization of
\(\iota_{\forall}(L_{A,B,C})\).

Finally, we consider a probabilistic refinement of the existential and universal
subsequence questions.  For a fixed binary pattern \(w\) and a nonempty language
\(L_{A,B,C}\), define \(p_w(A,B,C)\) to be
the probability that a uniformly chosen word from
\(L_{A,B,C}\) contains \(w\) as a subsequence.  The endpoint conditions
\(p_w(A,B,C)>0\) and \(p_w(A,B,C)=1\) are exactly the existential and universal
subsequence problems for the fixed pattern \(abab\).  The probability \(p_w(A,B,C)\)
therefore interpolates between the two quantifier extremes: instead of asking
whether at least one representative contains \(w\), or whether all representatives
contain \(w\), it asks for the density of representatives containing \(w\).

Exact computation of this probability is not expected to be easy in general, since
the universal subsequence problem, that is, testing \(p_w(A,B,C)=1\), is as hard as factoring integers already for the fixed pattern $w$.
Nevertheless, for every fixed word pattern \(w\) and fixed error tolerance
\(\varepsilon>0\), we give an additive approximation algorithm for
\(p_w(A,B,C)\) running in time polynomial in the bitlength of \(A,B,C\).  The proof
uses the correspondence between binary Parikh matrix/\(2\)-deck classes and partitions
inside an \(A\times B\) rectangle: each word in \(L_{A,B,C}\) corresponds to a
partition of area \(C\), 
in the form of Young diagrams (for more details, see Richomme's discussion of this connection~\cite{Richomme26}).
We combine this partition model with a polynomial bound on the number of
\(w\)-avoiding representatives, a lower bound for Gaussian binomial coefficients
when the input counts are not bounded by a constant, and exact counting procedures when one of
\(A,B,C,AB-C\) is bounded.  This yields a constant-error approximation scheme for
fixed word patterns, complementing the exact arithmetic characterizations of the
$\exists$- and $\forall$-universality indices.

We structure our paper as follows. The next section contains preliminaries for our work. Sections~\ref{sec:exist} and~\ref{sec:universal} are dedicated to the investigation of the $k$-universal existence and universality, while the last section introduces the concept of probability with respect to membership of subsequences and gives an approximation scheme for it. 

\section{Preliminaries}

Let $\Sigma=\{a<b\}$ be an ordered binary alphabet, and let \(\Sigma^*\) denote the language of all words over \(\Sigma\). Let $|w|$ denote the number of occurrences of symbols of $\Sigma$ in a word $w=w_1w_2\cdots w_{|w|}$ with $w_i\in\Sigma$, for $i\in[1,|w|],$ and by $w^R=w_{|w|}\cdots w_2w_1$ denote its reversal.
For words \(u,w\in\Sigma^*\), we write
\(u\preceq w\) if \(u\) occurs in \(w\) as a {\em subsequence}, i.e., if
there exist indices \(1\le i_1<\cdots<i_{|u|}\le |w|\) such that
\(u=w[i_1]\cdots w[i_{|u|}]\).
By $\Subseq(w)$ we denote the set of all subsequences occurring in $w.$
Let $|w|_u$ denote the number of occurrences of $u$ in $w$ as a subsequence, for any $u,w\in \Sigma^*$, where occurrences are allowed to overlap, and are distinct if they differ in at least one position, e.g., $|abb|_{ab}=2$. For $i<j$, by $[i,j]$ we denote the set of integers starting at $i$ and ending at $j$, inclusive. 

For any integer $s\geq 2$, let $\mathcal{M}_s$ denote the multiplicative monoid of $s \times s$ upper triangular matrices with nonnegative integer entries and $1$'s on the main diagonal.
The Parikh matrix mapping~\cite{mateescu2001sharpening} with respect to the ordered alphabet $\Sigma=\{a<b\}$, denoted by $\Psi$, is a morphism from $\Sigma^*$ into $\mathcal{M}_{3}$, where
\[
\Psi(a)= 
\left(\begin{smallmatrix}
    1 & 1 & 0\\
    0 & 1 & 0\\
    0 & 0 & 1
\end{smallmatrix}\right)
\qquad \text{ and }\qquad
\Psi(b)=
\left(\begin{smallmatrix}
    1 & 0 & 0\\
    0 & 1 & 1\\
    0 & 0 & 1
\end{smallmatrix}\right).
\]
As $\Psi$ is a morphism, it maps a word $w=a_{1}a_{2}\cdots a_{{\vert w\vert}}\in  \Sigma^*$ to the matrix $\Psi(w)= \Psi( a_{1}) 
\cdots \Psi( a_{{\vert w\vert}}).$
Matrices of the form $\Psi(w)$, for $w\in\Sigma^*$, are termed \emph{Parikh matrices}.
Hence, over the alphabet $\{a<b\}$, Parikh matrices are $3\times 3$ upper triangular matrix with 1's on the diagonal, regardless of the described word,
and the fundamental theorem of the Parikh matrix mapping~\cite{mateescu2001sharpening} shows that the upper part tracks the number of occurrences, as subsequences, of $a$, $b$ and $ab$ in the respective word, i.e., 
$|w|_{a}$ appears in position $(1,2)$, $|w|_{b}$ in position $(2,3)$, and $|w|_{ab}$ in position $(1,3).$

Note that when considering binary subsequences every pair of positions consisting of one \(a\) and one \(b\) contributes either to the total number of subsequences \({ab}\) or to \({ba}\). Therefore, we have that in every binary word $w$ the following equality holds 
\begin{equation}\label{eq:a_x_b}
    |w|_{ab}=|w|_a\cdot |w|_b-|w|_{ba}.
\end{equation}

The {\em $k$-deck} of a word is the multiset resulting from counting all occurrences of subsequences of length at most $k$ that the respective word contains. As such, one can define the language corresponding to a given $k$-deck as the language of all words sharing the respective multiset. Furthermore, when the underlying alphabet is binary, as a straightforward consequence of Eq.~\ref{eq:a_x_b}, both the Parikh matrices and the $2$-decks associated with a word contain the same information.

For a word $w\in\Sigma^*$, its 
equivalence class, with respect to the Parikh matrix mapping, is the language $L(M)=\{v\in \Sigma^* \mid \Psi(v)=M\}$ determined by its matrix $M=\Psi(w)$. We call such a language a {\em (binary) Parikh matrix language}, and for the rest of this work, for nonnegative integers \(A,B,C\), define
\begin{equation}\label{eq:L_abc}
    L_{A,B,C}
 =
 \{w\in\{a,b\}^*: |w|_a=A,\ |w|_b=B,\ |w|_{ab}=C\}.
\end{equation}
 We also denote $|w|_{ba}=D$ and have $D=AB-C.$ By definition, this equivalence class is identical to the one corresponding to the $2$-deck notion.
Moreover, 
it is straightforward that a language $L_{A,B,C}$ is nonempty if and only if $0\le C\le AB.$
Indeed, writing a word in the form
$w=b^{x_0}ab^{x_1}a\cdots ab^{x_A},$
we have
$\sum\limits_{i=0}^{A}x_i=B$ and $|w|_{ab}=\sum\limits_{i=0}^{A} i x_i.$
Every integer between \(0\) and \(AB\) is obtainable by distributing the
\(B\) letters \(b\) among the \(A+1\) gaps separating the occurrences of $a$.

For more details on Parikh matrices we refer to~\cite{mateescu2001sharpening,DingSalomaa2006,SalYu10}, while for $k$-decks we suggest~\cite{Kalashnik73,Scott97}.

A word \(w\in\Sigma^*\) is called {\em \(k\)-subsequence universal}, or simply
{\em \(k\)-universal}, over \(\Sigma\), if
\(
 \Sigma^k\subseteq \Subseq(w).
\)
Equivalently, every word of length \(k\) over \(\Sigma\) occurs in \(w\) as a
subsequence.
The {\em universality index of \(w\)}, denoted \(\iu(w)\), is the largest \(k\) such
that \(w\) is \(k\)-universal.

For a language \(L\subseteq \Sigma^*\), we use the two language-level notions
from the \(k\)-universality literature, see~\cite{Adamson2025105357}.
We say that \(L\) is {\em \(k\)-$\exists$-universal} if there exists \(w\in L\) such that
\(\iu(w)\ge k\) and define $\iexists(L)=\max\{k: L$ is $k$-$\exists$-universal$\}$.
We say that a nonempty \(L\) is {\em \(k\)-$\forall$-universal} if every \(w\in L\) satisfies
\(\iu(w)\ge k\) and define $\iforall(L)=\max\{k: L$ is $k$-$\forall$-universal$\}$.

As the universality index $\iforall(L)$ above is meaningless for empty languages, and one can easily decide the emptiness of Parikh matrix/2-deck languages over binary alphabets, from here on we assume that computing $\iforall(L)$ means computing it for a nonempty language.

The \emph{arch-factorization} of a word is a decomposition into non-overlapping factors starting from the left, such that each factor -- except for, possibly, the last one -- is the shortest possible, while containing each letter of the alphabet. The factorization is unique and its factors are called arches, with the ones containing all letters called complete arches. For example, over the alphabet $\{a,b\}$, the arch-factorization of $aabbabaa$ is $aab.ba.ba.a$, whereas that of $aabbb$ is $aab.bb$.
The following characterization of $k$-universality of binary words is a well-known argument based on the arch-factorization (see~\cite{dlt2020, Hebrard1991} for the general case); we shall use it in order to develop some of the tools we use in our analysis.

\begin{lemma}\label{lem:binary-arches}
For \(w\in\{a,b\}^*\) and \(k\ge 0\), the following are equivalent:
\begin{enumerate}[label=(\roman*)]
\item \(w\) is \(k\)-universal over \(\{a,b\}\);
\item \(w\) contains, as a subsequence, some word
$v\in\{ab,ba\}^k.$
\end{enumerate}
\end{lemma}
\begin{proof}
The implication $(ii)\Rightarrow(i)$ is immediate.  Every factor \(ab\) and
every factor \(ba\) of $v$ (the arches) contains both letters \(a\) and \(b\).  Obviously, from a
subsequence $v=s_1s_2\cdots s_k,$ where $s_i\in\{ab,ba\},$
one can obtain every possible word \(x_1x_2\cdots x_k\in\{a,b\}^{k}\) by choosing
the symbol for \(x_i\) from the \(i\)-th arch \(s_i\).  Hence \(v\), and
therefore also \(w\), are \(k\)-universal.

For $(i)\Rightarrow(ii)$, we use the standard greedy arch factorization, see~\cite{Hebrard1991}.  Since
\(w\) is \(k\)-universal, its arch-factorization contains at least \(k\)
complete arches, i.e., non-overlapping factors having both $a$ and $b$. Therefore, each of the arches contains either \(ab\) or \(ba\) as a subsequence; choosing one such
two-letter subsequence from each of the first \(k\) arches gives a subsequence
belonging to \(\{ab,ba\}^k\).
\end{proof}

From the lemma above we can directly conclude the following consequence.
\begin{corollary}[\cite{DaFlMaNo19}]\label{rem:shortest-universal}
The shortest \(k\)-universal binary words are exactly the words in $\{ab,ba\}^k.$
\end{corollary}



\section{$\exists$-universality}\label{sec:exist}

We first turn our attention to when \(L_{A,B,C}\) contains a \(k\)-universal word.
The next lemma is the main interval-realization tool.

\begin{lemma}\label{lem:supersequence-interval}
Let \(v\in\{a,b\}^*\), and fix integers \(A\ge |v|_a\) and \(B\ge |v|_b\).  Then the set of possible values of
\(|w|_{ab}\), over all words \(w\) with $|w|_{a}=A$ and $|w|_{b}=B$, and having $v$ as subsequence, i.e., $v\preceq w,$
is the integer interval 
$[|v|_{ab},\ AB-|v|_{ba}]\cap\mathbb Z.$
\end{lemma}

\begin{proof}
First, if \(v\preceq w\), then the embedded \(v\) contributes at least
\(|v|_{ab}\) pairs of type \(ab\).  Thus \(|w|_{ab}\ge |v|_{ab}\).  Similarly, the 
copy of \(v\) in \(w\) contributes at least \(|v|_{ba}\) pairs of type \(ba\).  Since Eq.~\ref{eq:a_x_b} holds for $w$, 
we get $|w|_{ab}\le AB-|v|_{ba}.$
So every attainable value lies in the stated interval.

Conversely, start with \(v\), and insert all extra
\(B-|v|_{b}\) letters \(b\) to the far left and all extra \(A-|v|_{a}\) letters \(a\) to
the far right.  The resulting word has \(A\) letters \(a\), \(B\) letters
\(b\), contains \(v\) as a subsequence, and has exactly \(|v|_{ab}\) subsequences
of type \(ab\).

Now move continuously to the opposite extreme, where all extra letters \(a\)
are placed to the far left and all extra letters \(b\) are placed to the far
right, keeping the original copy of \(v\) in its original relative order.  This
can be done by adjacent swaps \(ba\mapsto ab\), each involving at least one of the discussed
extra letters.  Such a swap preserves the chosen copy of \(v\) and increases
\(|w|_{ab}\) by exactly \(1\).  At the opposite extreme, no pair involving an
extra letter contributes to \(|w|_{ba}\), so the only remaining \(ba\)-pairs
are those already inside \(v\).  Hence the final value of \(|w|_{ab}\) is
\(AB-|v|_{ba}\).

Since at each step \(|w|_{ab}\) changes by one, every integer
between \(|v|_{ab}\) and \(AB-|v|_{ba}\) is realized.
\end{proof}



\begin{lemma}\label{lem:block-counts}
If \(t(v)\) denotes the number of length-$2$ factors starting at odd positions (blocks) equal to \(ab\) in \(v\in\{ab,ba\}^k\), then
$|v|_{ab}=\binom{k}{2}+t(v),$
and $|v|_{ba}=\binom{k}{2}+k-t(v).$
\end{lemma}
\begin{proof}
Each pair of distinct blocks \(i<j\) contributes one \(ab\)-pair: choose the
letter \(a\) from block \(i\) and the letter \(b\) from block \(j\).  There are
\(\binom{k}{2}\) such cross-block contributions.  In addition, every block \(ab\) contributes one \(ab\)-pair, while a block \(ba\) does
not.  Hence 
$|v|_{ab}=\binom{k}{2}+t(v).$
The formula for \(|v|_{ba}\) follows either by the same argument or from
$|v|_{ab}+|v|_{ba}=k^2$, that is, Eq.~\ref{eq:a_x_b}.
\end{proof}

\begin{theorem}
\label{thm:existential}
For every \(k\ge 0\), a nonempty language \(L_{A,B,C}\)
is $k$-$\exists$-universal
if and only if $A\ge k,$ $B\ge k,$
 $C\ge \binom{k}{2},$ and $AB-C\ge \binom{k}{2}.$
\end{theorem}
\begin{proof}
%
Suppose first that \(L_{A,B,C}\) is \(k\)-$\exists$-universal.  Then there exists
\(w\in L_{A,B,C}\) with \(\iu(w)\ge k\).  By
Lemma~\ref{lem:binary-arches}, \(w\) contains a subsequence
\(v\in\{ab,ba\}^k\).  Hence \(w\) contains at least \(k\) letters \(a\) and at
least \(k\) letters \(b\), so \(A\ge k\) and \(B\ge k\).

Since \(v\preceq w\), we have that $w$ also contains as subsequences at least as many pairs $ab$ and $ba$, as $v$ does. Thus, by Lemma~\ref{lem:block-counts}, we conclude that $C=|w|_{ab}\ge \binom{k}{2},$ and $|w|_{ba}\ge \binom{k}{2}.$

Conversely, assume $A\ge k$, $B\ge k$, $C\ge N$, and $D=AB-C\ge N,$
where \(N=\binom{k}{2}\).  We shall choose
\(v\in\{ab,ba\}^k\) such that \(v\preceq w\), for some
\(w\in L_{A,B,C}\).

By Lemmas~\ref{lem:block-counts} and~\ref{lem:supersequence-interval},
a word \(v\in\{ab,ba\}^k\) with \(t(v)\) blocks \(ab\) can be embedded in some
word \(w\) with parameters $A$, $B$, and $C$ if and only if
$N+t(v)\le C\le AB-(N+k-t(v)).$

Equivalently,
$N+k-D\le t(v)\le C-N.$
We need an integer \(t(v)\in\{0,\ldots,k\}\) satisfying this inequality.
Since, from the statement, \(D\ge N\), we have \(N+k-D\le k\), and since \(C\ge N\), we have
\(C-N\ge 0\).  Finally,
\[
 (C-N)-(N+k-D)
 =
 C+D-2N-k
 =
 AB-k^2.
\]
Because \(A\ge k\) and \(B\ge k\), we have \(AB\ge k^2\).  Hence the interval 
$[N+k-D,\ C-N]$
has a nonempty intersection with \(\{0,\ldots,k\}\).  Choose an integer \(t(v)\) in this intersection and
choose any word \(v\in\{ab,ba\}^k\) with exactly \(t(v)\) blocks \(ab\).  Then
Lemma~\ref{lem:supersequence-interval} gives a word \(w\in L_{A,B,C}\) with
\(v\preceq w\).  By Lemma~\ref{lem:binary-arches}, this word \(w\) is
\(k\)-universal.
\end{proof}

We can now summarize our results so far.
\begin{corollary}\label{cor:existential-index}
For nonempty \(L_{A,B,C}\), we have
$\iexists(L_{A,B,C})
=
\min\left\{
A,\ B,\ 
\left\lfloor
\frac{1+\sqrt{1+8\min(C,AB-C)}}{2}
\right\rfloor
\right\}.$
\end{corollary}

\begin{proof}
By Theorem~\ref{thm:existential}, the largest admissible \(k\) is bounded by
\(A\), by \(B\), and by the largest \(k\) satisfying
$\binom{k}{2}\le \min(C,AB-C).$
Solving \(k(k-1)/2\le M\), for \(M=\min(C,AB-C)\), gives
$k\le
 \left\lfloor
 \frac{1+\sqrt{1+8M}}{2}
 \right\rfloor.$
\end{proof}

Furthermore, we draw the following algorithmic consequences.
\begin{corollary}
Given nonnegative integers $A$, $B$, $C$, and $k$ in binary, one can decide in polynomial time whether \(L_{A,B,C}\) is \(k\)-$\exists$-universal.
When \(L_{A,B,C}\ne\emptyset\), we can compute \(\iota_{\exists}(L_{A,B,C})\) in polynomial time.

\end{corollary}

\begin{proof}
First we check whether \(0\le C\le AB\). If this test fails, then the input does not correspond to any word, i.e., \(L_{A,B,C}\) is empty, so the answer is negative. The decision test consists of integer comparisons involving \(A,B,C,AB-C\) and
\(k(k-1)/2\).  The index formula uses an integer square root of
\(1+8\min(C,AB-C)\), which is computable in polynomial time in the bitlength
of
the input.
\end{proof}

\section{$\forall$-universality}\label{sec:universal}

We now turn to the $\forall$-universality index.  The first observation is
that this index is always at most \(3\) for nonempty binary Parikh matrix
languages.

\begin{lemma}\label{lem:at-most-three}
There exists no nonempty \(4\)-$\forall$-universal language \(L_{A,B,C}\).  Consequently,
$\iforall(L_{A,B,C})\le 3.$
\end{lemma}

\begin{proof}
If \(A=0\) or \(B=0\), then no word in \(L_{A,B,C}\) is even
\(1\)-universal, so the claim is immediate.  

For binary words, a concrete characterization of lexicographically minimal (and maximal) representatives of a Parikh matrix language is given in~\cite[Theorem 3]{SalYu10}. These descriptions can also be deduced from the languages given in~\cite[page 9]{DingSalomaa2006}. 

In particular, the minimal representative is an element of the language $a^*b^*ab^*a^*$, while the maximal is of the form $b^*a^*ba^*b^*$. Since those two contain only words with at most $3$ arches, together with Lemma~\ref{lem:binary-arches}, we conclude.
%
%
%
\end{proof}

%

It is trivial that a nonempty language $L_{A,B,C}$ is $1$-$\forall$-universal if and only if every one of its words contains both symbols of the alphabet.

\begin{proposition}\label{prop:forall-one-two}
A nonempty language \(L_{A,B,C}\) is $2$-$\forall$-universal if and only if $A,B \geq 2$ and $\min(C,AB-C)>0$.
\end{proposition}

\begin{proof}
%
A binary word is \(2\)-universal if and only if it contains all of the four words $aa$, $ab$, $ba$, and $bb$ as subsequences.
The pattern \(aa\) occurs in every word of \(L_{A,B,C}\) if and only if
\(A\ge2\), and \(bb\) occurs in every word if and only if \(B\ge2\).  The
pattern \(ab\) occurs in every word if and only if \(C=|w|_{ab}>0\), while
\(ba\) occurs in every word if and only if \(|w|_{ba}>0\).
\end{proof}


For \(3\)-$\forall$-universality we need to understand when a representative of
\(L_{A,B,C}\) can avoid a fixed binary word of length \(3\).

\begin{lemma}\label{lem:avoidance-table}
%
Let \(L_{A,B,C}\) be nonempty. For each \(u\in\{a,b\}^3\), there exists a word $w$ that does not have $u$ as a subsequence, if and only if the following avoidability conditions hold (where \(D=AB-C\)):
\[
\begin{array}{|c||c|c|c|c|c|c|c|c|}
\hline
\text{pattern }u &aaa & bbb & aab & abb & baa & bba & aba & bab\\
\hline
\text{condition} &A<3
& B<3
& C\le B
& C\le A
& D\le B
& D\le A
& A=0\text{ or }A\mid C
& B=0\text{ or }B\mid C\\
\hline
\end{array}
\]
In particular, if \(A,B\ge3\), the last two conditions reduce to
\(A\mid C\) and \(B\mid C\), respectively.
\end{lemma}

\begin{proof}
The first two conditions are immediate: every word with at least three \(a\)'s
contains \(aaa\), and every word with at least three \(b\)'s contains \(bbb\).

For \(aab\), write a word as
$w=b^{x_0}ab^{x_1}a\cdots ab^{x_A}.$
The letters \(b\) within the $b^{x_i}$ blocks have exactly \(i\) letters \(a\) before them.
Thus \(w\) avoids \(aab\) if and only if every \(b\) has at most one \(a\)
before it, i.e., if and only if all \(b\)'s lie in the first two blocks.  This is
possible with total \(ab\)-count \(C\) if and only if \(C\le B\).  Hence the
third condition follows.

The condition for \(abb\) is dual.  A word avoids \(abb\) if and only if every
letter \(a\) has at most one \(b\) after it.  Since
$|w|_{ab}=\sum\limits_{\text{letters }a} \#\{b\text{ after this }a\},$
this is possible if and only if \(C\le A\).

Similarly, a word avoids \(baa\) if and only if every letter \(b\) has at most
one \(a\) after it.  Since \(D=|w|_{ba}\), this is possible if and only if
\(D\le B\).  A word avoids \(bba\) if and only if every letter \(a\) has at
most one \(b\) before it, which is possible if and only if \(D\le A\).

A word avoids \(aba\) if and only if
no \(b\) lies between two occurrences of \(a\).  Equivalently, every
\(b\) is either before all \(a\)'s or after all \(a\)'s.  If \(A=0\), this is
automatic.  If \(A>0\), such a word has the form
$b^{B-q}a^A b^q$,
for some \(0\le q\le B\), and its \(ab\)-count is \(qA\).  Hence such a word
exists if and only if \(A\mid C\).

The proof for \(bab\) is symmetric.  A word avoids \(bab\) if and only if
every \(a\) is either before all \(b\)'s or after all \(b\)'s.  If \(B=0\),
this is automatic.  If \(B>0\), such a word has the form
$a^p b^B a^{A-p}$,
and its \(ab\)-count is \(pB\).  Hence such a word exists if and only if
\(B\mid C\).
\end{proof}

\begin{theorem}\label{thm:forall-three}
A nonempty language \(L_{A,B,C}\) is $3$-$\forall$-universal if and only if
$A, B\ge 3$, 
$\min(C,AB-C)>\max(A,B)$,
and
$A\nmid C$ and $B\nmid C.$
\end{theorem}

\begin{proof}
A binary word is \(3\)-universal if and only if it contains all eight words in
\(\{a,b\}^3\).  Thus \(L_{A,B,C}\) is \(3\)-$\forall$-universal if and only if
no word in \(L_{A,B,C}\) avoids any of these eight patterns.

By Lemma~\ref{lem:avoidance-table}, avoidance of \(aaa\) is possible if and
only if \(A<3\), and avoidance of \(bbb\) is possible if and only if \(B<3\).
Therefore we must have \(A\ge3\) and \(B\ge3\).

Under this assumption, the remaining six avoidance conditions are
\[
 C\le B,\qquad C\le A,\qquad AB-C\le B,\qquad AB-C\le A,\qquad A\mid C,\qquad B\mid C.
\]
The negation of the first two is
$C>\max(A,B),$
the negation of the next two is
$AB-C>\max(A,B),$
and the negation of the last two is
$A\nmid C$ and $B\nmid C.$
This leads to the stated condition.
\end{proof}


Combining the preceding results yields the $\forall$-universality index.

\begin{theorem}\label{thm:forall-index}
For a nonempty language \(L_{A,B,C}\) we have 
\[
\iota_{\forall}(L_{A,B,C})=
\begin{cases}
3, & \text{if } A,B\ge 3,\ \min(C,AB-C)>\max(A,B),\
A\nmid C,\ B\nmid C,\\
2, & \text{if the previous condition fails, and }
A,B\ge2,\ \min(C,AB-C)>0,\\
1, & \text{if the previous conditions fail, and } A,B\ge1,\\
0, & \text{if } A=0 \text{ or } B=0.
\end{cases}
\]

\end{theorem}

\begin{proof}
By Lemma~\ref{lem:at-most-three}, the universal index is at most \(3\).
The condition for value at least \(3\) is exactly
Theorem~\ref{thm:forall-three}.  If this condition fails, the index is at most
\(2\), and Proposition~\ref{prop:forall-one-two} gives the condition for
being at least \(2\).  If that condition also fails, the index is at most \(1\). It is at least \(1\)
exactly when every word contains both letters, i.e., when \(A,B\ge 1\).
If \(A=0\) or \(B=0\), no word in the language is \(1\)-universal over the
fixed alphabet \(\{a,b\}\), so the value is \(0\).
\end{proof}

As a consequence we get the following algorithmic result.

\begin{corollary}
Given nonnegative integers $A$, $B$, $C$, and $k$ in binary, one can decide in polynomial time whether
\(L_{A,B,C}\) is \(k\)-$\forall$-universal.  One can also compute
\(\iforall(L_{A,B,C})\) in polynomial time.
\end{corollary}

\begin{proof}
First check nonemptiness, i.e., \(0\le C\le AB\). For an empty language, the answer is vacuously YES.  For nonempty languages,
Theorem~\ref{thm:forall-index} gives a constant number of arithmetic tests:
comparisons involving \(A,B,C,AB-C\), and divisibility tests \(A\mid C\) and
\(B\mid C\).  These are polynomial-time operations on binary integers.
\end{proof}


\section{An additive approximation scheme for fixed-pattern containment}\label{sec:ratio}
For a fixed word \(w\) and a finite language $L$, define the subsequence probability of \(w\) with respect to \(L\) as the ratio $p_w(L)=\frac{|\{u\in L \mid w\preceq u\}|}{|L|},$ that is, the probability that a uniformly randomly chosen word from \(L\) contains \(w\) as a subsequence. 

In our binary setting, we restrict notation and say that for a binary word $w$ and a nonempty language $L_{A,B,C}$, described via nonnegative integers $A, B, C$, we have
\[
p_w(A,B,C)=
\frac{|\{u\in L_{A,B,C} \mid w\preceq u\}|}{|L_{A,B,C}|}.
\]


The extremal cases of the problem considered in this section are the so-called $\exists$-\textsc{Subsequence} and $\forall$-\textsc{Subsequence} problems~\cite{FazekasKMMS24}, defined as follows:
\begin{itemize}
    \item $\exists$-\textsc{Subsequence}: For a language $L$ and a word $w$, decide whether there exists $u\in L$ such that $w\preceq u$.
    \item $\forall$-\textsc{Subsequence}: For a language $L$ and a word $w$, decide whether $w\preceq u$ holds, for all $u\in L$.
\end{itemize}

In terms of the ratio $p_w(L)$ defined above, the $\exists$-\textsc{Subsequence} and $\forall$-\textsc{Subsequence} problems ask to decide whether $p_w(L)>0$ and $p_w(L)=1$, respectively.
First we show that $\forall$-\textsc{Subsequence} is intractable in the case of binary 2-deck languages, directly implying that it is not expected that the ratio $p_w(A,B,C)$ is computable exactly in polynomial time unless factoring integers is easy.


We use a simple reduction showing that the complement of the
\(\forall\)-\textsc{Subsequence} problem (``is there a word in the language of a given Parikh matrix that does \emph{not} contain the given subsequence'') for binary Parikh matrix languages is
\textsc{Factoring*}-hard, even for the fixed pattern \(abab\). Here
\(\textsc{Factoring*}\) denotes the following interval-divisor problem: given
positive integers \(U,V,N\), with \(U\le V\le N\), decide whether \(N\) has a
divisor in the interval \([U,V]\). It is easy to see that the usual decision version of factoring integers (does the input integer $N$ have a prime factor at most $V$?) reduces trivially to $\textsc{Factoring*}$. Therefore $\textsc{Factoring*}$ is not expected to have a polynomial-time algorithm, since factoring integers is a candidate for an NP-intermediate problem (in NP$\ \cap\ $coNP, but believed to be outside P).
\begin{theorem}
The complement of $\forall$-\textsc{Subsequence} for binary Parikh matrix
languages is \(\textsc{Factoring*}\)-hard, even for the fixed pattern \(abab\).
\end{theorem}
\begin{proof}
Set
$C=N$, $A=V$, $B=\lfloor\frac{N}{U}\rfloor,$
and $w=abab$, so the instance of 
``\emph{do all $v\in L_{A,B,C}$ contain $w$ as a subsequence?}'' created for the \textsc{Factoring*} instance $(U,V,N)$ is \(\langle L_{A,B,C},abab\rangle.
\)
The resulting language \(L_{A,B,C}\) may be empty, but this is not a problem; the universal statement becomes vacuously true, and the corresponding instance of the complement problem is negative.

A binary word avoids \(abab\) as a subsequence if and only if its print, that is, the word obtained by replacing every maximal unary block by a single occurrence of its letter,
avoids
\(abab\). Equivalently, it has the form \(b^*a^*b^*a^*\), possibly allowing some of the
four unary blocks to be empty. Thus, if $abab$ is not a common subsequence of all words in $L_{A,B,C}$, then one of those words has the form $b^*a^*b^*a^*.$
If $N=dq=C$, where $U\le d\le V$ and $q\le \lfloor\frac{N}{U}\rfloor$, then the word $b^{B-q}a^{d}b^{q}a^{A-d}$ is in $L_{A,B,C}$ and it does not contain $abab$ as subsequence.

Conversely, suppose some $x\in L_{A,B,C}$ does not contain $abab$, i.e., it has the form $b^{B-q}a^db^qa^{A-d}$, for some $d\le A$ and $q\le B$. This means that the number of $ab$'s in the word is exactly $dq$, so $C=dq$. Moreover, $q\le \lfloor\frac{N}{U}\rfloor$, so $d=C/q=N/q\geq N/\lfloor\frac{N}{U}\rfloor$, which gives $d\ge U$. Finally, $d\le V$, since it is the length of an $a$-block in a word with $V$ occurrences of $a$. Note that the division by $q$ above is always valid, since \(dq=C=N>0\).
\end{proof}

We showed that the complement of $\forall$-\textsc{Subsequence} is \textsc{Factoring*}-hard, which is at least as hard as factoring integers into primes. To decide factoring, one would test \(L_{A,B,C}=\emptyset\) first and if yes, answer negatively to the factoring instance. If the language is not empty, testing $p_{abab}(A,B,C)=1$ is equivalent to the $\forall$-\textsc{Subsequence} problem. Thus, an exact polynomial-time algorithm for the former would yield a polynomial-time algorithm for
\(\textsc{Factoring*}\). This motivates the additive approximation result below.

In what follows, we prove that, for every fixed \(w\) and error tolerance $\varepsilon$, the ratio \(p_w(A,B,C)\) admits an additive
polynomial-time approximation scheme in the binary encoding of
\(A,B,C\). When $A,B,C$ are fixed and no confusion arises, we streamline the notation further and use only $p_w$ for the ratio with respect to the fixed language $L_{A,B,C}$.

The endpoint cases \(p_w>0\) and \(p_w=1\) correspond respectively to
the \(\exists\)-\textsc{Subsequence} and \(\forall\)-\textsc{Subsequence} problems discussed above.  Thus \(p_w\) interpolates between those two endpoint
questions.  As we have just seen, the \(\forall\)-\textsc{Subsequence} problem can be hard even
for fixed patterns, but that hardness only distinguishes \(p_w=1\) from
\(p_w\le 1-1/|L_{A,B,C}|\), an exponentially small gap.  Hence it does
not preclude constant-error additive approximation.


In line with the geometric approach in~\cite{Richomme26}, we note that every word \(u\in L_{A,B,C}\) corresponds uniquely to a partition
$\lambda=(\lambda_1,\ldots,\lambda_A)\subseteq \{0,\dots,B\}^A$
of area \(C\) (that is, an integer partition of the number $C$), where
$\lambda_i=\#\{b$'s after the $i$th $a$ in $u\}$.
Thus
$B\ge \lambda_1\ge\cdots\ge\lambda_A\ge0,$
 where
$|\lambda|=\lambda_1+\cdots+\lambda_A=C.$

Consequently we have 
\[
|L_{A,B,C}|
=\#\{\lambda\subseteq \{0,\ldots,B\}^A \mid |\lambda|=C\}
=[q^C]\binom{A+B}{A}_q,
\]
where 
\[
\binom{A+B}{A}_q=\prod_{i=1}^A\frac{1-q^{B+i}}{1-q^i}
\]
is the standard Gaussian binomial coefficient, whose correspondence with the number of partitions of the above type is a classical result (see~\cite[Section 1.7]{Stanley}). In other words, this formula gives the coefficient of $q^C$ in this polynomial, which is the number of non-increasing partitions with area $C$ in a bounded box $A\times B$.
The binary
Parikh matrix equivalence classes are equivalently described by this
partition model; e.g., see the binary characterizations of
Fossé--Richomme~\cite{FosRic04} and later uses of the partition viewpoint.


In what follows, a binary word \(u\) is called \(w\)-avoiding if \(w\not\preceq u\).
Let
$Avoid_w(A,B,C)=|\{u\in L_{A,B,C} \mid w\not\preceq u\}|$ be the size of the subset of $L_{A,B,C}$ consisting of $w$-avoiding words.
We first bound \(Avoid_w\) by a fixed polynomial in the smaller of the two subsequences' pair counts $|ab|=C$ and $|ba|=AB-C.$ 
In what follows, a \emph{template} is an alternating binary word, that is, $\sigma_1 \sigma_2\cdots \sigma_r$, for some $r$, where $\sigma_i,\sigma_{i+1}\in\{a,b\}$ and $\sigma_i\neq \sigma_{i+1}$, for $1\le i<r$. A word $u$ matches the template, if there exist positive integers $x_1,x_2,\ldots,x_r$ such that $u=\sigma_1^{x_1} \sigma_2^{x_2}\cdots \sigma_r^{x_r}$.

\begin{lemma}\label{lem:bounded-run-avoiders}
For every fixed binary pattern \(w\), there exists a constant \(r_w\)
such that every $w$-avoiding binary word matches one of finitely many
templates
$\sigma_1\sigma_2\cdots\sigma_r$ with 
$r\le r_w$,
where \(\sigma_i, \sigma_{i+1}\in\{a,b\}\) with $\sigma_i\neq \sigma_{i+1}$, for $1\le i<r$, and
the finite set of templates depends only on \(w\).
\end{lemma}
\begin{proof}
To see this, start by using the usual greedy subsequence automaton for \(w=w_1\cdots w_\ell\).
The state \(q\) means that \(w_1\cdots w_q\) has been greedily matched.
A word avoids \(w\) if and only if the automaton never reaches state \(\ell\).
The automaton can advance at most \(\ell-1\) times. Between two consecutive advances, as well as before the first and after the last, the automaton can only read repetitions of one letter.
Hence, every \(w\)-avoiding word has at most \(2\ell-1\) maximal unary blocks, and therefore matches one of finitely many alternating templates.
\end{proof}

\begin{lemma}
\label{lem:avoider-poly-bound}
For every fixed binary pattern \(w\), there exist constants \(m_w,e_w\)
such that, for all nonempty languages \(L_{A,B,C}\), we have
$Avoid_w(A,B,C)
\le
m_w\cdot \left(\min\{C,AB-C\}+1\right)^{e_w}$.
\end{lemma}

\begin{proof}
Fix one template $\sigma_1\sigma_2\cdots \sigma_r$ from the finite family established by Lemma~\ref{lem:bounded-run-avoiders}. We count the words of \(L_{A,B,C}\) matching that template, that is, words of the form $u=\sigma_1^{x_1}\sigma_2^{x_2}\cdots \sigma_r^{x_r}\in L_{A,B,C},$ 
where all \(x_i\ge 1\). For each such \(u\) we have
$C=|u|_{ab}=\sum\limits_{i<j,\ \sigma_i=a,\ \sigma_j=b} x_i x_j$.

First assume that $C\leq AB-C$, and we bound the number of such words in terms of \(C\). If a block \(i\) participates in at least one
summand in this expression, then, since all block lengths are positive,
that summand is at least \(x_i\). Hence \(x_i\le C\). Thus every
interacting block has length at most \(C\).
The only blocks not participating in any \(ab\)-summand are the possible
initial \(b\)-run and the possible final \(a\)-run. Moreover, adjacent symbols in the template are distinct, there is at most one initial \(b\)-block and one final \(a\)-block. Once the interacting block
lengths have been chosen, these outer lengths are uniquely determined by
$\sum\limits_{\sigma_i=a}x_i=A$ and, respectively,
$\sum\limits_{\sigma_i=b}x_i=B$.
Therefore this template gives at most \((C+1)^r\) possible words.

If \(AB-C<C\), apply the same argument to the reversed words. Reversal sends
\(L_{A,B,C}\) bijectively to \(L_{A,B,AB-C}\), and \(w\preceq u\) is equivalent
to \(w^R\preceq u^R\). Since \(w^R\) is fixed, the same polynomial bound applies
with \(AB-C\) in place of \(C\).
Thus the number of $w$-avoiding words is polynomial in \(AB-C\) in this case.

It remains to sum these bounds over the finite template families from Lemma~\ref{lem:bounded-run-avoiders}. We choose $e_w=\max\{r_w,r_{w^R}\}$ (the larger of the constants from Lemma~\ref{lem:bounded-run-avoiders} for $w$ and its reverse) to be the bound on the number of blocks, and $m_w$ the maximum of the numbers of templates associated with $w$ and $w^R$.
\end{proof}


The counts of letters or the subsequences $ab$ or $ba$ being bounded by a constant (boundary regimes) implies languages with few words or very simple structure. We first give a simple lower bound for coefficients of Gaussian binomials when we are not in those boundary regimes. 

\begin{lemma}\label{lem:gaussian-lower}
For every fixed integer \(E\geq 0\), 
there exists \(R_E\) such that whenever
$A,B,C,AB-C\ge R_E,$
we have
$[q^C]\binom{A+B}{A}_q
\ge
\left(\min\{C,AB-C\}+1\right)^E.$
\end{lemma}


\begin{proof}
Let
$P_{A,B}(t)=[q^t]\binom{A+B}{A}_q$.
Thus \(P_{A,B}(t)\) is the number of partitions of \(t\) whose Young diagrams
fit inside an \(A\times B\) rectangle. The coefficients \(P_{A,B}(t)\) are unimodal by Sylvester's classical result (for a more modern reference, see~\cite{OHARA199029}).

Let $s=\min\{C,AB-C\}$. First we recall that the coefficients of a Gaussian binomial coefficient are symmetric, that is,
\[
[q^t]\binom{A+B}{A}_q = [q^{AB-t}]\binom{A+B}{A}_{q}.
\]
This symmetry follows from the fact that complementing a Young diagram of area $t$ in the $A\times B$ rectangle and rotating it by $180^\circ$ gives a Young diagram enclosed in the same rectangle, but with area $AB-t$.
By this symmetry, it is enough to prove a lower bound
for \(P_{A,B}(s)\), where \(s\le AB/2\). Indeed, if \(s=C\), then
\(P_{A,B}(C)=P_{A,B}(s)\), while if \(s=AB-C\), then
\(P_{A,B}(C)=P_{A,B}(AB-C)=P_{A,B}(s)\).

Fix $H=2E+4$.
We choose \(R_E\) sufficiently large so that \(R_E\ge (2H)^2\) and, for all
\(s\ge R_E\),
\[
\frac{s^{E+2}}{(2H)^H}\ge (s+1)^{E+1}.
\]
The reason for requiring $E+1$ in the exponent here instead of just $E$ is due to a cancellation argument at the end. Now, what we need to show is that, whenever \(A,B,s\ge R_E\) and \(s\le AB/2\), we have
\[
\sum_{t=0}^s P_{A,B}(t)\ge \frac{s^{E+2}}{(2H)^H}.
\]

The reason for considering the cumulative sum is the following. If we can find
a subrectangle of the \(A\times B\) rectangle whose area is at most \(s\), then
every partition fitting inside this subrectangle has area at most \(s\), and
therefore contributes to one of the coefficients \(P_{A,B}(0),\ldots,
P_{A,B}(s)\). Thus a large enough subrectangle of area at most \(s\) gives a
lower bound on the cumulative sum. We shall choose such a subrectangle so that
one of its sides is of order \(\sqrt{s}\), while the other has fixed size
\(H\). This will give polynomially many partitions, with exponent controlled
by \(H\).

Set $M=\min\{A,B\}$ and $N=\max\{A,B\}$.
We want to construct a subrectangle \(\alpha\times\beta\subseteq M\times N\) such that
$\alpha\beta\le s$, $\min\{\alpha,\beta\}\ge H$,
and $\max\{\alpha,\beta\}\ge \sqrt{s}/2$.
There are two cases.

First suppose that \(M\ge \sqrt{s}\) and we take
$\alpha=\beta=\lfloor \sqrt{s}\rfloor$.
This square fits inside \(M\times N\), since
\(\alpha,\beta\le M\le N\), and its area is at most \(s\). Since
\(s\ge R_E\ge (2H)^2\), both sides are at least \(H\). Also
\(\lfloor\sqrt{s}\rfloor\ge \sqrt{s}/2\), so the larger side is at least
\(\sqrt{s}/2\).

Now suppose that \(M<\sqrt{s}\). We take
$\alpha=M$ and $\beta=\left\lfloor\frac{s}{M}\right\rfloor$.
Then \(\alpha\le M\), and \(\alpha\beta\le s\). Moreover, since
\(s\le MN/2\), we have
$\frac{s}{M}\le \frac{N}{2}$,
and hence \(\beta\le N\). Therefore the rectangle
\(\alpha\times\beta\) fits inside \(M\times N\). We also have
\(\alpha=M\ge R_E\ge H\). Since \(M<\sqrt{s}\), we get
\(s/M>\sqrt{s}\), and hence
\[
\beta=\left\lfloor\frac{s}{M}\right\rfloor\ge \sqrt{s}-1
       \ge \frac{\sqrt{s}}{2},
\]
where the last inequality holds because \(s\ge R_E\) is large enough. In
particular, \(\beta\ge H\). Thus the required properties hold in this case
as well.

Since \(M\times N\) is either \(A\times B\) or \(B\times A\), the constructed rectangle fits inside the \(A\times B\) rectangle, possibly after transposing it when $B<A$. In the latter case, transposition gives an area-preserving bijection between partitions fitting inside the \(\alpha\times\beta\) rectangle and partitions fitting inside the \(\beta\times\alpha\) rectangle. Hence, after possibly interchanging
\(\alpha\) and \(\beta\), we may assume that
\(
\alpha\le A
\)
and
\(
\beta\le B.
\)

We have therefore found, in all cases, a subrectangle
\(\alpha\times\beta\subseteq A\times B\) with area at most \(s\), both sides
at least \(H\), and larger side at least \(\sqrt{s}/2\). Consequently,
\[
\sum_{t=0}^s P_{A,B}(t)
\ge
\#\{(\lambda_1,\ldots,\lambda_\alpha) \mid
\beta\ge\lambda_1\ge\cdots\ge\lambda_\alpha\ge0\}
=
\binom{\alpha+\beta}{\alpha}.
\]
Let
$J=\max\{\alpha,\beta\}$.
Since the rectangle \(\alpha\times\beta\) has one side \(J\) and the other side
at least \(H\), it contains, up to transposition, an \(H\times J\) rectangle.
Hence
\[
\binom{\alpha+\beta}{\alpha}
\ge
\binom{H+J}{H}
=
\prod_{i=1}^H \frac{J+i}{i}
\ge
\left(\frac{J}{H}\right)^H .
\]
Using \(J\ge \sqrt{s}/2\), we obtain
\[
\sum_{t=0}^s P_{A,B}(t)
\ge
\left(\frac{\sqrt{s}}{2H}\right)^H
=
\frac{s^{H/2}}{(2H)^H}.
\]
Since \(H=2E+4\), we have \(H/2=E+2\), and therefore
\[
\sum_{t=0}^s P_{A,B}(t)
\ge
\frac{s^{E+2}}{(2H)^H}.
\]

Finally, since \(s\le AB/2\), symmetry and unimodality imply that
\(P_{A,B}(s)\) is the largest coefficient among
\(P_{A,B}(0),P_{A,B}(1),\ldots,P_{A,B}(s)\). Hence
\[
(s+1)P_{A,B}(s)\ge \sum_{t=0}^s P_{A,B}(t),
\qquad
\text{ and therefore }\qquad
P_{A,B}(s)
\ge
\frac{1}{s+1}\cdot \frac{s^{E+2}}{(2H)^H}.
\]
By the choice of \(R_E\), the right-hand side is at least \((s+1)^E\).
Therefore, as required,
$P_{A,B}(s)\ge (s+1)^E$.
Since \(s=\min\{C,AB-C\}\), symmetry gives \(P_{A,B}(C)=P_{A,B}(s)\).
Therefore
\[
[q^C]\binom{A+B}{A}_q
=
P_{A,B}(C)
=
P_{A,B}(s)
\ge
(\min\{C,AB-C\}+1)^E.\qedhere
\]
%
\end{proof}

Next we show that, if each of $A,B,C,AB-C$ is large enough, the containment probability
of a fixed pattern is either 0 or arbitrarily close to 1.

\begin{lemma}\label{lem:interior-dominance}
Fix a binary pattern \(w\) and \(\varepsilon>0\).  There exists
\(R=R(w,\varepsilon)\) such that if
$A,B,C,AB-C\ge R$
and the language \(L_{A,B,C}\) contains at least one word having \(w\) as subsequence, then
$p_w(A,B,C)\ge 1-\varepsilon$.
\end{lemma}

\begin{proof}
Let $m_w$ and $e_w$ be the constants from
Lemma~\ref{lem:avoider-poly-bound}.  Choose an integer \(E>e_w\).
By Lemma~\ref{lem:gaussian-lower}, 
whenever
$A,B,C,AB-C\ge R_E$
we have that
\[
|L_{A,B,C}|
=
[q^C]\binom{A+B}{A}_q
\ge
\left(\min\{C,AB-C\}+1\right)^E.
\]
On the other hand,
$Avoid_w(A,B,C)
\le
m_w\left(\min\{C,AB-C\}+1\right)^{e_w}$.
Therefore
\begin{equation}\label{eq:gaussian}
\frac{Avoid_w(A,B,C)}{|L_{A,B,C}|}
\le
\frac{m_w}
{\left(\min\{C,AB-C\}+1\right)^{E-e_w}}.
\end{equation}
Choose \(R\ge R_E\) sufficiently large that
\[
\frac{m_w}{(R+1)^{E-e_w}}\le \varepsilon.
\]
Since $m_w$, $e_w$ and $E$ are constants, such a choice is possible, and $A,B,C,AB-C\ge R$ guarantees that the right-hand side of inequality~(\ref{eq:gaussian}) above is at most
\(\varepsilon\).  Hence
\[
p_w(A,B,C)=1-\frac{Avoid_w(A,B,C)}{|L_{A,B,C}|}
\ge 1-\varepsilon.
\qedhere\]
\end{proof}


Let \(R\) be fixed.  We show that \(p_w(A,B,C)\) can be computed exactly
in time
$f_w(R)\cdot \operatorname{poly}(\log A+\log B+\log C)$
whenever at least one of
$A$, $B$, $C$, or $AB-C$
is at most \(R\). Note that here $f_w(R)$ denotes a function which in general is not necessarily polynomial in $R$.

\begin{lemma}\label{lem:boundary-exact}
Fix a binary pattern \(w\) and an integer \(R\). There is an algorithm which
computes \(p_w(A,B,C)\) in time
\(f_w(R)\cdot \operatorname{poly}(\log A+\log B+\log C)\), whenever
\(L_{A,B,C}\) is nonempty and at least one of
\(A,B,C,AB-C\) is at most \(R\).
\end{lemma}
\begin{proof}
If any of the input counts are $0$, then the probability can be computed trivially. Now assume the counts are positive.\\
\noindent{\bf Small area \(C\le R\).}
Every partition of area \(C\) fits inside a \(C\times C\) square.  Thus
we enumerate all partitions \(\lambda\) of \(C\), discard those
with more than \(A\) parts (rows) or largest part greater than \(B\), and test
whether the corresponding word contains \(w\) as subsequence.

If \(\lambda\) has \(r\le C\) nonzero parts and largest part \(s\le C\),
then the corresponding word has the form
$b^{B-s}\,U_\lambda\,a^{A-r}$,
where 
\(
U_\lambda=a b^{\lambda_1-\lambda_2}
a b^{\lambda_2-\lambda_3}
\cdots
a b^{\lambda_{r-1}-\lambda_r}
a b^{\lambda_r}
\) that has at most \(C\) occurrences of \(a\) and at most
\(C\) occurrences of \(b\).  For testing containment of the fixed pattern
\(w\), the outer blocks may be capped at \(|w|\):
\[
b^{B-s}\quad\leadsto\quad b^{\min\{B-s,|w|\}},
\qquad
a^{A-r}\quad\leadsto\quad a^{\min\{A-r,|w|\}}.
\]
Thus each test is performed on a word of length \(\bigo(C+|w|)\).

Since the number \(p(C)\) of partitions of \(C\) is
$p(C)\in\exp(\bigo(\sqrt{C}))$,
the small-area case is exact in time
$\exp(\bigo(\sqrt R))\operatorname{poly}(|w|,\log A,\log B)$.

\noindent{\bf Small coarea \(AB-C\le R\).}
By reversing the word we have
$|u^R|_{ab}=|u|_{ba}=AB-C$,
and
$w\preceq  u$ is equivalent to
$w^R\preceq  u^R$.
Hence
$p_w(A,B,C)=p_{w^R}(A,B,AB-C)$,
which reduces to the small-area case.

\noindent{\bf Small height \(A\le R\).}
Write
\[
u_\lambda =
b^{B-\lambda_1}a b^{\lambda_1-\lambda_2}a\cdots
a b^{\lambda_A},
\]
where \(B\ge\lambda_1\ge\cdots\ge\lambda_A\ge0\) and
\(\sum_i\lambda_i=C\).  In 
\(
p_w(A,B,C)=
\frac{|\{u\in L_{A,B,C} \mid w\preceq u\}|}{|L_{A,B,C}|},
\) 
the denominator is the number of integer points in
this fixed-dimensional polytope.  For the numerator, run the deterministic
greedy subsequence automaton for the fixed word \(w\) on the input \(u_\lambda\) above.
Consider one of its unary $b$-blocks of length, say, $t$. From any state of the automaton, reading $b^t$ can advance the greedy match by at most $|w|$ states. Hence, for determining the resulting state, it is enough to distinguish the cases 
\(
t=0,1,\ldots,|w|-1
\)
and $t\ge |w|$.
Therefore, since both
\(A\le R\) and \(w\) are fixed, the condition \(w\preceq u_\lambda\) is a
finite 
disjoint union of systems
of linear inequalities in the variables \(\lambda_1,\ldots,\lambda_A\).
Therefore, both the numerator and the denominator can be computed by fixed-dimensional
integer-point counting, for instance by Barvinok's algorithm~\cite{Barvinok}, in time
\(f_w(R)\cdot \operatorname{poly}(\log A+\log B+\log C)\).


\noindent{\bf Small width \(B\le R\).}
Let \(\tau\) be the morphism exchanging \(a\) and \(b\).
The map \(u\mapsto \tau(u^R)\) is a bijection from \(L_{A,B,C}\) to
\(L_{B,A,C}\), since reversing changes \(ab\)-pairs into \(ba\)-pairs and then
exchanging the letters changes them back into \(ab\)-pairs. Moreover,
\(w\preceq u\) if and only if \(\tau(w^R)\preceq \tau(u^R)\). Thus the small-width
case reduces to the small-height case for the fixed pattern \(\tau(w^R)\).
\end{proof}

Finally, we combine the preceding lemmas to get the main result.

\begin{theorem}
\label{thm:fixed-pattern-ptas}
Fix a binary pattern \(w\) and an error tolerance \(\varepsilon>0\). There is an
algorithm which, on input nonnegative integers \(A,B,C\) with
\(L_{A,B,C}\neq\emptyset\), outputs a rational number \(p\) satisfying
\[
|p-p_w(A,B,C)|\le \varepsilon,
\]
and runs in time polynomial in the bitlength of \(A,B,C\), where the polynomial
may depend on \(w\) and \(\varepsilon\).

\end{theorem}

\begin{proof}
Choose \(R=R(w,\varepsilon)\) as in
Lemma~\ref{lem:interior-dominance}.

First test whether there exists any word in \(L_{A,B,C}\) containing
\(w\) as subsequence.  For binary Parikh matrices, this 
is decidable in polynomial time by the interval criterion of  Lemma~\ref{lem:supersequence-interval} for
the possible \(ab\)-count after embedding the pattern. 
If no such word
exists, then \(p_w(A,B,C)=0\), and the algorithm outputs \(0\).

Otherwise, if one of
$A$, $B$, $C$, or $AB-C$
is at most \(R\), compute \(p_w(A,B,C)\) exactly by the boundary
algorithms from Lemma~\ref{lem:boundary-exact}.

If none of these four quantities is at most \(R\), then
$A,B,C,AB-C>R$.
By Lemma~\ref{lem:interior-dominance}, we have
$p_w(A,B,C)\ge 1-\varepsilon$.
The algorithm outputs \(1\), which has an additive error at most
\(\varepsilon\).

The algorithms in the boundary regimes (constant-bounded counts) run in time depending on \(R\) and \(w\), times a
polynomial in the bitlength of \(A,B,C\).  
For fixed \(w\) and fixed \(\varepsilon\), the number \(R=R(w,\varepsilon)\)
is a constant. Hence the boundary algorithms run in
\(f_{w}(R)\cdot \operatorname{poly}(\log A+\log B+\log C)\) time, and the
interior case is decided by the constant number of arithmetic tests above.
Thus the total running time is polynomial in the input bitlength.
%
\end{proof}

\section{Acknowledgements}
The authors used OpenAI’s ChatGPT 5.5 to assist with drafting this manuscript, editing it for clarity, and consolidating the preliminary proof ideas of Section~\ref{sec:ratio}. All arguments were checked, revised, and approved by the authors, who take full responsibility for the content.

\bibliographystyle{eptcs}
\bibliography{AFL_bib}

\end{document}